\documentclass{ecai} 

\usepackage{latexsym}
\usepackage{amssymb}
\usepackage{amsmath}
\usepackage{amsthm}
\usepackage{booktabs}
\usepackage{enumitem}
\usepackage{graphicx}
\usepackage{color}
\usepackage{mathtools}
\usepackage{thm-restate}

\newtheorem{theorem}{Theorem}
\newtheorem{lemma}[theorem]{Lemma}
\newtheorem{corollary}[theorem]{Corollary}
\newtheorem{proposition}[theorem]{Proposition}

\newtheorem{problem}{Problem}

\newcommand{\BibTeX}{B\kern-.05em{\sc i\kern-.025em b}\kern-.08em\TeX}

\begin{document}


\begin{frontmatter}


\paperid{2268} 


\title{Existence of MMS Allocations of Mixed Manna}


\author[A]{\fnms{Kevin}~\snm{Hsu}\orcid{0000-0002-8932-978X}\thanks{Corresponding Author. Email: kevinhsu996@gmail.com.}}

\address[A]{University of Victoria}


\begin{abstract}
Maximin share (MMS) allocations are a popular relaxation of envy-free allocations that have received wide attention in the fair division of indivisible items. Although MMS allocations of goods can fail to exist, previous work has found conditions under which they exist. Specifically, MMS allocations of goods exist whenever $m \leq n+5$, and this bound is tight in the sense that they can fail to exist when $m = n+6$. The techniques used to obtain these results do not apply to the mixed manna setting, leaving the question of whether similar results hold for the general setting. This paper addresses this by introducing new techniques to handle these settings. In particular, we are able to answer this question completely for the chores setting, and partially for the mixed manna setting. An agent $i$ is a {\em chores agent} if it considers every item to be a chore and a {\em non-negative agent} if its MMS guarantee is non-negative. In this paper, we prove that an MMS allocation exists as long as $m \leq n+5$ and either (i) every agent is a chores agent, or (ii) there exists a non-negative agent. In addition, for $n \leq 3$, we also prove that an MMS allocation exists as long as $m \leq n+5$, regardless of the types of agents. To the best of our knowledge, these are the first non-trivial results pertaining to the existence of exact MMS allocations in the mixed manna setting.
\end{abstract}

\end{frontmatter}


\section{Introduction}

Fair division is the study of fairness in relation to the allocation of resources and responsibilities and the mechanisms for finding fair allocations. Since resources are often limited, methods of finding fair allocations have been sought since times immemorial. The theory of fair division depends immensely on the nature of the items being divided, such as their desirability and divisibility. For example, half of a movie ticket is a useless piece of paper, so movie tickets are considered to be indivisible. On the other hand, an electricity bill can be split arbitrarily between housemates, so it is divisible but undesirable. This paper is concerned with the division of indivisible items. In this setting, highly desirable classical fairness notions such as envy-freeness and proportionality are oftentimes unattainable.  For example, when dividing a single good between two agents, any allocation would leave one agent envious of the other. This has led to the introduction of a number of alternative fairness notions. The most popular among these include envy-freeness up to one good (EF1) and  maximin share (MMS) introduced by \citet{budish2011combinatorial}, and envy-freeness up to any good (EFX) introduced by \citet{gourves2014near} under the name of {\em near envy-freeness} and also by \citet{caragiannis2019unreasonable}.

In this paper, we focus on MMS allocations. Informally, an allocation is said to be MMS if it provides each agent with its {\em MMS guarantee}, which is the utility that an agent would receive if it is able to decide how the items are divided into bundles, but must receive the least valuable bundle. Although MMS allocations can fail to exist \citep{KPW16}, they have been shown to exist in some restricted settings when the agents have additive utility functions. The vast majority of the existing results have focused on goods division when the agents have additive utility functions. In particular, \citet{BL16} showed that when dividing goods (i.e. items with non-negative utility), MMS allocations exist if $n=2$ or $m \leq n+3$, where $n, m$ are the numbers of agents and goods, respectively. Later, \citet{KPW16} improved this bound to $m \leq n+4$ while also showing that for all $n \geq 3$, there exists an $n$-agent $(3n+4)$-good fair division instance that admits no MMS allocations. A tight result was obtained by \citet{feige2021tight} who improved the bound to $m \leq n+5$ and showed that this is the best possible upper bound on $m$ of the form $n+k$ where $k \in \mathbb{Z}$ by exhibiting an instance with $(n, m) = (3, 9)$ for which an MMS allocation fails to exist.

The common underlying approach to much of the above work relies on reducing a given instance of the fair division problem into a smaller sub-instance with the usage of {\em reduction rules}. Essentially, given an instance of the fair division problem, a reduction rule $R$ allocates some of the items to some of the agents, and returns the resulting sub-instance. If the agents eliminated by $R$ are satisfied with respect to their respective MMS guarantees, and if the MMS guarantees of the remaining agents are preserved (i.e. do not decrease after applying $R$), then we can inductively consider the smaller sub-instance. The above mentioned work that led to the $m \leq n+5$ bound makes use of a reduction rule that allocates a single good to a single agent, which relies on the following fact.

\begin{proposition}\label{prop:delete_one_good}\citep{BL16}
	Let $I$ be an instance of the fair division problem. Suppose all agents have additive utility functions and $I$ contains only goods. Let $I_R$ be an instance obtained from $I$ by deleting an arbitrary agent and an arbitrary good. Then, for any agent $i$ in $I_R$, we have
	$$u_i^\text{MMS}(I) \leq u_i^\text{MMS}(I_R)$$
	where $u_i^\text{MMS}(I)$ and $u_i^\text{MMS}(I_R)$ denote the MMS guarantee of agent $i$ in instances $I$ and $I_R$, respectively. \qed
\end{proposition}

Unfortunately for chores and mixed manna settings, deleting an arbitrary agent and an arbitrary item does not always provide the same guarantee (an example is provided in Appendix \ref{appen:A}). This raises the question of whether MMS allocations exist in the chores and mixed manna settings under the same condition of $m \leq n+5$. The purpose of this paper is to address this question.

In order to obtain our new results, we used a combination of existing techniques based on reduction rules and a new technique based on modifying the utility functions of the agents. At first glance, modifying utility functions may seem counterintuitive because doing so changes the problem instance. Even if the altered instance admits an MMS allocation, such an allocation is not necessarily MMS for the original instance. However, we introduce a way of modifying utility functions such that any MMS allocation for the altered instance can be used to find an MMS allocation for the original instance (see Proposition \ref{prop:mimic}).

\subsection{Our Contribution}

An agent is said to be a {\em goods agent} if every item provides non-negative utility to it. Similarly, an agent is said to be a {\em chores agent} if every item provides non-positive utility to it. If an agent is neither a goods agent nor a chores agent, it is said to be a {\em mixed agent}. Given an instance $I$ of the fair division problem and an agent $i$, we define the {\em MMS guarantee} of agent $i$ to be the utility (in the view of agent $i$) of a minimum-utility bundle among all possible allocations. An agent is {\em non-negative} if its MMS guarantee is non-negative.

Our main contribution in this paper is the following theorem.

\begin{theorem}\label{thm:main}
	Let $I$ be an instance of the fair division problem. Suppose $m \leq n+5$ and at least one of the following conditions hold.
	\begin{enumerate}
		\item $I$ contains $n \leq 3$ agents.
		\item $I$ contains a non-negative agent.
		\item $I$ contains only chores agents.
	\end{enumerate}
	Then, $I$ admits an MMS allocation.
\end{theorem}

Notably, Theorem \ref{thm:main}(3) is an existence result for the chores setting. On the other hand, the class of goods agents is strictly contained in the class of non-negative agents, because goods agents are necessarily non-negative. Thus, Theorem \ref{thm:main}(2) strictly generalizes the $m \leq n+5$ result concerning goods division due to \citet{feige2021tight}.

\subsection{Related Work}

This paper is concerned with establishing the existence of exact MMS allocations by considering restricted types of instances, specifically by limiting the number of items relative to the number of agents. Below, we discuss some other recent work related to MMS allocations. For a more comprehensive survey of other fairness concepts, we refer the reader to the survey by \citet{amanatidis2022fair}.

\textbf{Limited number of items.}  A recent paper by \citet{hummel2023lower} in a similar vein as this paper showed that the tightness of the bound $m \leq n+5$ can be overcome by considering asymptotic behaviour. Specifically, Hummel showed that for any integer $c > 0$, there exists an integer $n_c$ such that any instance with $n \geq n_c$ agents and $m \leq n + c$ goods admits an MMS allocation. Furthermore, the same paper also showed that for $n \neq 3$, an MMS allocation of goods exists as long as $m \leq n+6$.

\textbf{Restricted utility functions.} Rather than limiting the number of items, domain restriction by way of restricting the types of utility functions has also been explored. For example, \citet{barman2021existence} showed that MMS allocations of goods exist if all agents have matroid-rank utility functions and can be found in polynomial-time. Interestingly, the MMS allocations that they found are also Pareto-optimal. In the same paper, they complemented this positive result by showing that MMS allocations can fail to exist if agents have XOS utilities, which are a generalization of matroid-rank utilities. On the other hand, \citet{hosseini2023fairly} showed that for lexicographic utilities, MMS allocations of mixed manna exist and can be computed in polynomial-time.

\textbf{Approximations.} Since exact MMS allocations are not guaranteed to exist, much effort has been dedicated to finding approximate MMS allocations. An allocation is $\alpha$-MMS if it provides each agent an $\alpha$ multiplicative factor of its MMS guarantee.
For goods, the current state-of-the-art result is due to \citet{akrami2024breaking}, who showed that an $(\frac{3}{4} + \frac{3}{3836})$-MMS allocation exists. On the negative side, \citet{feige2021tight} showed that it is impossible to achieve an approximation factor of $39/40$ even for the case of 3 agents. For chores division, the current best-known approximation is due to \citet{huang2023reduction}, who showed an approximation factor of $13/11 + \epsilon$ for the general case and an improved factor of $15/13$ for the case of 3 agents. On the negative side, \citet{feige2021tight} showed that it is impossible to achieve an approximation factor of $44/43$ even for the case of 3 agents. In contrast to both goods and chores division, \citet{kulkarni2021indivisible} showed that in the mixed manna setting, it is NP-hard to find an $\alpha$-MMS allocation for any constant factor $\alpha \in (0, 1]$ even when a solution exists when $\alpha = 1$.

\subsection{Organization}

In the following, Section 2 lays out the preliminary foundation. In Section 3, we prove the different parts of Theorem \ref{thm:main}. In Section 4, we conclude the paper with a brief discussion on future directions.

\section{Preliminaries}\label{sec:pre}

In this section, we formally define the model of fair division and MMS allocations. Then, we define reduction rules and introduce the reduction rules that we use.

\subsection{The Model}

An {\em instance} $I$ of the fair division problem is a tuple $(\mathcal{N}, \mathcal{M}, \mathbf{u})$ where $\mathcal{N} = [n]$ is a set of $n$ agents, $\mathcal{M} = \{o_1, o_2, \dots, o_m\}$ is a set of $m$ items, and $\mathbf{u}$ is a collection of $n$ utility functions $u_i:\mathcal{P}(\mathcal{M}) \rightarrow \mathbb{R}$ corresponding to the agents. An agent's utility function represents its preferences by assigning a numerical value to each subset of $\mathcal{M}$. For the case of a single item $o_j$, we will write $u_i(o_j)$ instead of $u_i(\{o_j\})$ for brevity. An {\em allocation} $\pi$ is an ordered $n$-tuple $(\pi_1, \pi_2, \dots, \pi_n)$ where each $\pi_i$ is a subset of $\mathcal{M}$ and no pair of distinct subsets $\pi_i, \pi_j$ intersect, and the union of all subsets $\pi_i$ is equal to $\mathcal{M}$. We refer to the subsets $\pi_i$ as {\em bundles}. A utility function is {\em additive} if for any pair of disjoint subsets $S, T \subseteq \mathcal{M}$, we have $u_i(S \cup T) = u_i(S) + u_i(T)$. We assume all utility functions are additive.

\subsection{Maximin Share Allocations}
We now formally introduce the notion of maximin share allocations. Fix an instance $I = (\mathcal{N}, \mathcal{M}, \mathbf{u})$ of the fair division problem. For a positive integer $n$, we use $\Pi_n(\mathcal{M})$ to denote the set of allocations $\pi = (\pi_1, \pi_2, \dots, \pi_n)$ of $\mathcal{M}$ among $n$ agents. We define the {\em maximin share guarantee} $u_i^\text{MMS}(I)$ of each agent $i$ as follows
$$
u_i^\text{MMS}(I) \coloneqq \max_{\pi \in \Pi_n(\mathcal{M})} \min_{k \in [n]} u_i(\pi_k)
$$
This is the utility that agent $i$ would receive if it decides the allocation but must receive the least valuable bundle. We say that a bundle $\pi_i$ {\em satisfies} an agent $i$ if $u_i(\pi_i) \geq u_i^\text{MMS}(I)$. An allocation $\pi = (\pi_1, \pi_2, \dots, \pi_n)$ is said to be {\em maximin share} (MMS) if for each $i$, the bundle $\pi_i$ satisfies the agent $i$. If every bundle of an allocation $\pi$ satisfies agent $i$, then we say $\pi$ is {\em MMS for agent} $i$. Clearly, for any agent $i$, there exists some allocation that is MMS for $i$ (consider any allocation that maximizes the utility of the least valuable bundle in the perspective of agent $i$), but such an allocation is not necessarily unique.

We say that an instance has {\em same-order preference} (SOP) if $u_i(o_1) \geq u_i(o_2) \geq \dots \geq u_i(o_m)$ for all agents $i$. Any instance $I$ that does not have SOP can be transformed into one that has SOP by choosing, for each utility function $u_i$, a permutation $\sigma:[m] \rightarrow [m]$ such that $u_i(o_{\sigma(1)}) \geq u_i(o_{\sigma(2)}) \geq \dots \geq u_i(o_{\sigma(m)})$ and replacing $u_i$ with a new utility function $u_i'$ defined by $u_i'(o_j) \coloneqq u_i(o_{\sigma(j)})$ for each item $o_j$. We call the resulting instance the {\em SOP instance corresponding to} $I$ and denote it by $I^\text{SOP}$. Observe the MMS guarantee of each agent is the same in $I$ as in $I^\text{SOP}$. This is because the MMS guarantee of an agent depends only on the number of agents and the item utilities, and does not depend on the order of the item utilities.

SOP instances are important because in a sense, they are the hardest instances for the purpose of finding an MMS allocation. This allows us to focus our attention on them.

\begin{proposition}\label{prop:sop}\citep{BL16}
	Let $I$ be an instance of the fair division problem and $I^\text{SOP}$ be the SOP instance corresponding to $I$. If $I^\text{SOP}$ admits an MMS allocation, then $I$ does as well. \qed
\end{proposition}

In the rest of this paper, we assume all instances have SOP in light of Proposition \ref{prop:sop}.

Our approach to finding MMS allocations is based on induction and requires the following technical result. In the case when there are exactly 2 agents, \citet{BL16} showed that MMS allocations can be found by a divide-and-choose algorithm.

\begin{proposition}\label{prop:2agents}
	\citep{BL16} Any 2-agent instance admits an MMS allocation. \qed
\end{proposition}

A particularly useful observation that can be made is that introducing dummy items with zero utility to a given instance does not affect whether or not an MMS allocation exists. This is because the MMS guarantee of each agent is unaffected by the addition of such dummy items. This allows us to assume $m = n+5$ without loss of generality, by adding an appropriate number of dummy items if $m < n+5$. We will often make use of this assumption when working with instances for which $m \leq n+5$.

\subsection{Types of Agents}

We distinguish between agents by their utility functions and MMS guarantees. An agent $i$ is said to be a {\em goods agent} if $u_i(o_j) \geq 0$ for each item $o_j \in \mathcal{M}$, and a {\em chores agent} if $u_i(o_j) \leq 0$ for each item $o_j \in \mathcal{M}$. It is possible for an agent to be neither a goods agent nor a chores agent. Such agents are called {\em mixed agents}. We also classify agents depending on the values of their MMS guarantees. An agent is said to be {\em positive}, {\em non-positive}, {\em negative}, or {\em non-negative} if its MMS guarantee is positive, non-positive, negative, or non-negative, respectively.


\subsection{Reduction Rules}

Let $I = (\mathcal{N}, \mathcal{M}, \textbf{u})$ be an instance of the fair division problem. It is useful to reduce $I$ to a smaller sub-instance $I_R$ by allocating a subset of items to a subset of agents. Formally, a {\em reduction rule} $R$ is a function $R$ that maps $I$ to a pair $(I_R, \pi)$ where

\begin{enumerate}
	\item $\pi$ is an allocation of a nonempty subset of items $\mathcal{M}' \subseteq \mathcal{M}$ to a nonempty subset of agents $\mathcal{N}' \subseteq \mathcal{N}$; and
	\item $I_R = (\mathcal{N}_R, \mathcal{M}_R, \textbf{u}_R)$ is the instance of the fair division problem obtained from $I$ by deleting the items in $\mathcal{M}'$ and the agents in $\mathcal{N}'$.
\end{enumerate}

We say a reduction rule $R$ {\em preserves the MMS guarantee} of an agent $i \in \mathcal{N}_R$ if $u_i^\text{MMS}(I_R) \geq u_i^\text{MMS}(I)$. We say that $R$ {\em satisfies} an agent $i \in \mathcal{N} \setminus \mathcal{N}_R$ if the bundle that $\pi$ allocates to agent $i$ satisfies agent $i$. If $R$ preserves the MMS guarantee of every agent $i \in \mathcal{N}_R$ and satisfies every agent $i \in \mathcal{N} \setminus \mathcal{N}_R$, we call $R$ a {\em valid reduction rule}. Valid reduction rules are useful for finding MMS allocations. To preserve consistency between the agent and item indices in $I$ and $I_R$, we assume that $R$ always deletes agents and items with the greatest indices, starting from agent $n$ and item $o_m$. This assumption can be made without loss of generality by permuting these indices.

\begin{proposition}\label{prop:validreduction}
	Let $I$ be an instance and $R$ be a valid reduction rule. If the sub-instance $I_R$ obtained from $I$ by applying $R$ admits an MMS allocation, then $I$ admits an MMS allocation.
\end{proposition}
\begin{proof}
	Let $\pi$ be the allocation that is implicitly defined by $R$ and $\pi'$ be an MMS allocation for the instance $I_R$. Clearly, applying both $\pi$ and $\pi'$ to the instance $I$ is an MMS allocation for $I$.
\end{proof}

First, we present two simple reduction rules.

\begin{proposition}\label{prop:1reduction_preserve}
	Let $I$ be an instance and $\pi = (\pi_1, \pi_2, \dots, \pi_n)$ be an allocation that is MMS for an agent $i$. The following two reduction rules preserve the MMS guarantee of agent $i$.
	\begin{enumerate}
		\item Allocating an item $o_a$ to an agent $j \neq i$ if $\pi$ contains a bundle $\pi_\ell$ with $u_i(o_a) \leq u_i(\pi_\ell)$.
		\item Allocating two items $o_a, o_b$ to an agent $j \neq i$ if $u_i(\{o_a, o_b\}) \leq u_i^\text{MMS}(I)$.
	\end{enumerate}
\end{proposition}
\begin{proof}
	(1): Without loss of generality, assume $o_a \in \pi_n$. If $\pi_\ell = \pi_n$, then we have $u_i(o_a) \leq u_i(\pi_n)$. Otherwise, we can relabel the bundles of $\pi$ so that $\pi_\ell = \pi_{n-1}$ and $u_i(o_a) \leq u_i(\pi_\ell) = u_i(\pi_{n-1})$. In either case, we have $u_i(o_a) \leq \max(u_i(\pi_{n-1}), u_i(\pi_n))$. We claim the allocation $\pi' = (\pi_1, \pi_2, \dots, \pi_{n-2}, \pi_{n-1} \cup \pi_n \setminus \{o_a\})$ certifies that $u_i^\text{MMS}(I_R) \geq u_i^\text{MMS}(I)$. Since $u_i$ is additive, we have
	\begin{align*}
		u_i(\pi_{n-1} \cup \pi_n \setminus \{o_a\}) &= u_i(\pi_{n-1}) + u_i(\pi_n) - u_i(o_a) \\
		&\geq u_i(\pi_{n-1}) + u_i(\pi_n) \\
		&- \max(u_i(\pi_{n-1}), u_i(\pi_n)) \\
		&= \min(u_i(\pi_{n-1}), u_i(\pi_n)) \\
		&\geq u_i^\text{MMS}(I)
	\end{align*}
	It follows that the utility of each bundle of $\pi'$ is at least $u_i^\text{MMS}(I)$, so $u_i^\text{MMS}(I_R) \geq u_i^\text{MMS}(I)$.
	
	(2): Without loss of generality, assume $o_a, o_b \in \pi_{n-1} \cup \pi_n$. We claim that $\pi' = (\pi_1, \pi_2, \dots, \pi_{n-2}, \pi_{n-1} \cup \pi_n \setminus \{o_a, o_b\})$ certifies that $u_i^\text{MMS}(I_R) \geq u_i^\text{MMS}(I)$. Since $u_i(\pi_n) \geq u_i^\text{MMS}(I) \geq u_i(\{o_a, o_b\})$, we have
	\begin{align*}
		u_i(\pi_{n-1} \cup \pi_n \setminus \{o_a, o_b\}) &= u_i(\pi_{n-1}) + u_i(\pi_n) - u_i(\{o_a, o_b\}) \\
		&\geq u_i(\pi_{n-1}) \\
		&\geq u_i^\text{MMS}(I)
	\end{align*}
	It follows that the utility of each bundle of $\pi'$ is at least $u_i^\text{MMS}(I)$, so $u_i^\text{MMS}(I_R) \geq u_i^\text{MMS}(I)$.
\end{proof}

Using Proposition \ref{prop:1reduction_preserve}, we can derive two useful valid reduction rules when the allocations that are MMS for each agent have certain properties. The first valid reduction rule is as follows.

\begin{lemma}\label{lemma:1reduction}
	Let $I$ be an instance and suppose that for each agent $i$, there exists an allocation that is MMS for agent $i$ and contains a singleton bundle. Then, there exists a valid reduction rule that allocates a single item to a single agent.
\end{lemma}
\begin{proof}
	For each agent $i$, let $\pi^i$ denote an allocation that is MMS for agent $i$ and contains a singleton bundle. Let $o_k$ denote the item with maximum index $k$ such that $\{o_k\}$ is a bundle in some allocation $\pi^j$. We show that the reduction rule $R$ that allocates $o_k$ to agent $j$ is a valid reduction rule. Since $\{o_k\}$ is a bundle in $\pi^j$, we have $u_j(o_k) \geq u_j^\text{MMS}(I)$, so $R$ satisfies agent $j$. On the other hand, for any agent $i \neq j$, there exists a singleton bundle $\pi_\ell$ in $\pi^i$ for which $u_i(o_k) \leq u_i(\pi_\ell)$ by the choice of $o_k$. By Proposition \ref{prop:1reduction_preserve}(1), $R$ preserves the MMS guarantee of agent $i$. Thus, $R$ is a valid reduction rule.
\end{proof}

The second valid reduction rule that we will show requires the following technical result, which can be understood as an iterated version of Proposition \ref{prop:1reduction_preserve}.

\begin{proposition}\label{prop:atmost2reduction}
	Let $I$ be an instance, $X$ be a set of $k$ agents, and $Y$ be a set of $k$ disjoint bundles each of size 1 or 2. Suppose that for some agent $j \notin X$, we have $u_j(\pi_s) \leq u_j^\text{MMS}(I)$ for each bundle $\pi_s$ in $Y$. Let $R$ be any reduction rule that allocates the bundles in $Y$ to agents in $X$ so that each agent in $X$ receives exactly one bundle in $Y$. Then, $R$ preserves the MMS guarantee of agent $j$.
\end{proposition}
\begin{proof}
	We prove this by induction on $k$ by repeatedly applying Proposition \ref{prop:1reduction_preserve}. If $k=0$, then there is nothing to show. Suppose $k=1$ and let $\pi_s$ be the unique bundle in $Y$. Let $\pi$ be an allocation that is MMS for agent $j$ and $\pi_\ell$ be any bundle in $\pi$. By the choice of $\pi_\ell$, we have $u_j(\pi_s) \leq u_j^\text{MMS}(I) \leq u_j(\pi_\ell)$. If $\pi_s$ has size 1, then allocating $\pi_s$ to the agent in $X$ preserves the MMS guarantee of agent $j$ by Proposition \ref{prop:1reduction_preserve}(1). Otherwise if $\pi_s$ has size 2, then allocating $\pi_s$ to the agent in $X$ preserves the MMS guarantee of agent $j$ by Proposition \ref{prop:1reduction_preserve}(2). Assume the proposition holds for all $k < k'$ for some $k'$. We proceed to show that it holds for $k = k'$ as well.
	
	Suppose $R$ allocates some bundle $\pi_s$ in $Y$ to some agent $i$ in $X$. We can view $R$ as being the composition of two reduction rules $R_1$ and $R_2$, where $R_1$ allocates $\pi_s$ to $i$ and $R_2$ allocates $Y \setminus \pi_s$ to $X \setminus \{i\}$ in the same way as $R$. Let $I_1$ be the instance obtained from $I$ by applying $R_1$, and $I_2$ be the instance obtained from $I_1$ by applying $R_2$. According to the induction hypothesis, $R_1$ preserves the MMS guarantee of agent $j$, so $u_j^\text{MMS}(I_1) \geq u_j^\text{MMS}(I)$. Similarly, the induction hypothesis implies that $R_2$ preserves the MMS guarantee of agent $j$ (when applied to the instance $I_1$), so $u_j^\text{MMS}(I_2) \geq u_j^\text{MMS}(I_1)$. Together, we have $u_j^\text{MMS}(I_2) \geq u_j^\text{MMS}(I)$, so applying $R_1$ followed by $R_2$ on the instance $I$ preserves the MMS guarantee of agent $j$. The proposition follows from the fact that applying $R_1$ followed by $R_2$ on the instance $I$ is equivalent to applying $R$ on the instance $I$.
\end{proof}

Using the above proposition, we can derive the following valid reduction rule.

\begin{lemma}\label{lemma:atmost2reduction}
	Let $I$ be an instance and suppose that for some agent $i$, there exists an allocation $\pi$ that is MMS for $i$ and contains $n-1$ bundles each of size 1 or 2. Then, $I$ admits an MMS allocation or there exists a valid reduction rule that allocates to $k$ agents at least $k$ items for some $k>0$.
\end{lemma}

\begin{proof}
	We begin by defining a bipartite graph using the allocation $\pi$. Let $G(\pi) = (A \cup B, E)$ be the bipartite graph where $A$ is the set of agents in the instance $I$ and $B$ is the set of bundles in the allocation $\pi$. For each agent $j$ and each bundle $\pi_k$, the graph $G(\pi)$ contains the edge $\{j, \pi_k\}$ if and only if the bundle $\pi_k$ satisfies the agent $j$. If $G(\pi)$ contains a perfect matching, then allocating each bundle to the agent it is matched with yields an MMS allocation, so we are done.
	
	Assume $G(\pi)$ does not contain a perfect matching. By Hall's theorem, there exists a set $Y \subseteq B$ such that $|N(Y)| < |Y|$, where $N(Y) \coloneqq \{j \in A: \{j, \pi_k\} \in E \text{ for some } \pi_k \in Y\}$ is the neighborhood of $Y$. We choose such a set $Y$ with minimal size, that is, $|N(Y)| < |Y|$ and no smaller set $Y'$ has the property that $|N(Y')| < |Y'|$. Note that $|Y| \geq 2$ because every vertex of $Y$ is adjacent to at least one vertex of $X$, namely, the vertex representing agent $i$ for whom $\pi$ is MMS. Since $\pi$ is MMS for agent $i$, the vertex $i$ is adjacent to every vertex of $B$. In particular, $Y$ does not contain an isolated vertex, so we have $0 < |N(Y)| < |Y|$. Hence, $|Y| \geq 2$.
	
	Since $\pi$ contains $n-1$ bundles each of size 1 or 2, at most one bundle has size not equal to 1 or 2. In particular, at most one bundle of $Y$ has size not equal to 1 or 2, so it is possible to delete one bundle from $Y$ to obtain a subset $Y' \subset Y$ such that $|Y'| = |Y| - 1 \geq 1$ such that each bundle of $Y'$ has size 1 or 2. Moreover, by the minimality of $Y$ and Hall's theorem, the subgraph of $G(\pi)$ induced by $Y'$ and $N(Y')$ contains a perfect matching $M$. We claim the reduction rule $R$ that allocates bundles in $Y'$ to agents in $N(Y')$ according to the matching $M$ is a valid reduction rule that allocates to $k$ agents at least $k$ items for some $k > 0$.
	
	Since $|Y'| \geq 1$ and each bundle of $Y'$ has size 1 or 2, the number of items $R$ allocates is at least the number of agents $R$ allocates items to. Hence, $R$ allocates to $k$ agents at least $k$ items for some $k > 0$. Moreover, $R$ satisfies the agents in $N(Y')$ because $M$ matches each of these agents to a bundle that satisfies that agent. It remains to show that $R$ preserves the MMS guarantees of the agents not in $N(Y')$. Let $j$ be an agent not in $N(Y')$. By construction, $Y'$ is a set of disjoint bundles each of size 1 or 2. Moreover, we have $u_j(\pi_k) < u_j^\text{MMS}(I)$ for each bundle $\pi_k$ of $Y'$ by the definition of $G(\pi)$ because $i$ is not in $N(Y')$. Thus, Proposition \ref{prop:atmost2reduction} implies that $R$ preserves the MMS guarantee of agent $j$, so $R$ is valid.
\end{proof}

We remark that the idea involving bundles of sizes 1 and 2 and Hall's theorem has also been applied independently by Hummel in the proof of Lemma 23 of \citep{hummel2023lower}, in which a version of Lemma \ref{lemma:atmost2reduction} for goods is implicitly obtained.

\section{Existence of MMS Allocations}

In this section, we prove Theorem \ref{thm:main}. Each condition of Theorem \ref{thm:main} require a different technique, so we divide the proof into different subsections.

\subsection{$I$ contains $n \leq 3$ agents}

The proof of Theorem \ref{thm:main} relies on applying valid reduction rules to reduce a given instance to a smaller instance. This requires the base cases for the induction to first be shown separately. The $n=2$ case has been settled by Proposition \ref{prop:2agents}. In this subsection, we additionally show the $n=3$ case.

The outline of the proof is as follows. For each agent $i \in [3]$ in an instance $I$, we fix an allocation $\pi^i$ that is MMS for $i$. For each pair of allocations $\pi^i, \pi^j$, we define a special bipartite graph $G(\pi^i, \pi^j)$. By analyzing this graph, we show that whenever some $\pi^i$ contains a bundle of size at most one or two 2-bundles (i.e. bundle of size 2), then $I$ admits an MMS allocation. This leaves the remaining case in which each of the three allocations $\pi^i$ contains one 2-bundle and two 3-bundles. In this case, we will reduce the instance to a 3-agent 9-item instance for which each agent $i$, there is an allocation that is MMS for $i$ and contains three 3-bundles. Finally, we will use a proposition due to \citet{feige2021tight} that implies that such an instance admits an MMS allocation.

We now introduce the special bipartite graph. For each agent $i \in [3]$, we fix an allocation $\pi^i$ that is MMS for $i$. For each pair of agents $i \neq j$, we define the bipartite graph $G(\pi^i, \pi^j) = (A \cup B, E)$ as follows. The vertex sets $A, B$ consist of the bundles of $\pi^i, \pi^j$, respectively. The edge set $E$ contains the edge $\{\pi^i_x, \pi^j_y\}$ if and only if the bundles $\pi^i_x, \pi^j_y$ are disjoint.

\begin{proposition}\label{prop:edge}
	Let $I$ be a 3-agent instance. For each agent $i$, let $\pi^i = (\pi^i_1, \pi^i_2, \pi^i_3)$ be an allocation that is MMS for agent $i$. Suppose $G(\pi^i, \pi^j)$ contains an edge $\{\pi^i_x, \pi^j_y\}$ and there exists a bundle $\pi^i_z \neq \pi^i_x$ such that $u_i(\pi^i_z) \geq u_i(\pi^j_y)$. Then, $I$ admits an MMS allocation.
\end{proposition}
\begin{proof}
	Without loss of generality, assume $G(\pi^1, \pi^2)$ contains the edge $\{\pi^1_1, \pi^2_1\}$ and $u_1(\pi^1_2) \geq u_1(\pi^2_1)$. Consider the allocation $\pi = (\pi^1_1, \pi^2_1, \mathcal{M} \setminus (\pi^1_1 \cup \pi^2_1))$. This allocation is well-defined because $\pi^1_1$ is disjoint with $\pi^2_1$. We have
	\begin{align*}
		u_1(\mathcal{M} \setminus (\pi^1_1 \cup \pi^2_1)) &= u_1(\mathcal{M} \setminus \pi^1_1) - u_1(\pi^2_1) \\
		&= u_1(\pi^1_2 \cup \pi^1_3) - u_1(\pi^2_1) \\
		&\geq u_1(\pi^1_2 \cup \pi^1_3) - u_1(\pi^1_2) \\
		&= u_1(\pi^1_3) \geq u_1^\text{MMS}(I)
	\end{align*}
	So, the first and third bundles of $\pi$ both satisfy agent 1. On the other hand, since the bundle $\pi_1^2$ of $\pi$ is also a bundle of $\pi^2$, the additivity of $u_2$ implies that one of the bundles $\pi_1^1$ and $\mathcal{M} \setminus (\pi^1_1 \cup \pi^2_1)$ of $\pi$ satisfies agent 2. Without loss of generality, assume $\pi_1^1$ satisfies agent 2. Hence, the first and second bundles of $\pi$ both satisfy agent 2. Finally, at least one of the three bundles of $\pi$ satisfies agent 3. Let agent 3 pick a bundle in $\pi$ satisfying them first. Regardless of which bundle agent 3 picks, it is clearly possible to assign the remaining two bundles to agents 1 and 2 in a way that satisfies them both.
\end{proof}

\begin{corollary}\label{cor:2edges}
	Let $I$ be a 3-agent instance. For each agent $i$, let $\pi^i$ be an allocation that is MMS for agent $i$. If for a pair of agents $i \neq j$, the graph $G(\pi^i, \pi^j)$ contains two edges, then $I$ admits an MMS allocation. In particular, if any $\pi^i$ contains a bundle of size at most one, then $I$ admits an MMS allocation.
\end{corollary}
\begin{proof}
	Without loss of generality, assume $G(\pi^1, \pi^2)$ contains two edges.  Suppose the two edges share an endvertex. Again without loss of generality, assume $\{\pi^1_1, \pi^2_1\}$ and $\{\pi^1_1, \pi^2_2\}$ are edges. By Proposition \ref{prop:edge}, if $u_1(\pi^1_2) \geq u_1(\pi^2_1)$ or if $u_1(\pi^1_2) \geq u_1(\pi^2_2)$, then $I$ admits an MMS allocation. Otherwise, $u_1(\pi^2_1) > u_1(\pi^1_2)$ and $u_1(\pi^2_2) > u_1(\pi^1_2)$, so the two bundles $\pi^2_1$ and $\pi^2_2$ of the allocation $\pi^2$ both satisfy agent 1. Using the picking order 3, 1, 2 on the $\pi^2$ (i.e. let agent 3 pick a bundle satisfying them first, then agent 1, then agent 2) results in an MMS allocation.
	
	Otherwise, the two edges do not share an endvertex. Without loss of generality, assume $\{\pi^1_1, \pi^2_1\}$ and $\{\pi^1_2, \pi^2_2\}$ are edges. Proposition \ref{prop:edge} implies that if $u_1(\pi^1_2) \geq u_1(\pi^2_1)$ or if $u_1(\pi^1_1) \geq u_1(\pi^2_2)$, then $I$ admits an MMS allocation. Otherwise, we have $u_1(\pi^2_1) > u_1(\pi^1_2)$ and $u_1(\pi^2_2) > u_1(\pi^1_1)$. Again, $\pi^2_1$ and $\pi^2_2$ both satisfy agent 1. Using the picking order 3, 1, 2 on the $\pi^2$ results in an MMS allocation.
	
	Suppose some $\pi^i$ contains a bundle of size at most one, say $\pi^i_1$. For any $j \neq i$, the bundle $\pi^i_1$ can only intersect with at most one bundle of $\pi^j$ because it contains at most one item. Equivalently, the bundle $\pi^i_1$ is disjoint with at least two bundles of $\pi^j$. Hence, the graph $G(\pi^i, \pi^j)$ contains two edges, so $I$ admits an MMS allocation.
\end{proof}

We also make use of the following proposition from \citet{feige2021tight}. Since its proof depends only on the additivity of the utility functions and the fact that $I$ has SOP, whether the items are goods or chores makes no difference to its validity. For completeness, we include its proof in Appendix \ref{appen:B}.

\begin{restatable}{proposition}{propthreenine}\label{prop:39}
	(Proposition 23 in \citep{feige2021tight}) Let $I$ be a 3-agent 9-item instance. If each $\pi^i$ contains three 3-bundles, then $I$ admits an MMS allocation.
\end{restatable}

\begin{theorem}\label{thm:38}
	Let $I$ be an instance of the fair division. Suppose $n=3$ and $m=8$. Then, $I$ admits an MMS allocation.
\end{theorem}
\begin{proof}
	For each agent $i$, let $\pi^i$ be an allocation that is MMS for $i$. If some $\pi^i$ contains a bundle of size at most one, then we are done by Corollary \ref{cor:2edges}. Otherwise, each $\pi^i$ contains only bundles of size at least two. Note that because $m=8$, the number of 2-bundles each $\pi^i$ contains is at most two.
	
	We distinguish two cases. In the first case, some $\pi^i$ contains two 2-bundles. Clearly, $\pi^i$ contains two 2-bundles and one 4-bundle. For any $\pi^j$, each of the two 2-bundles of $\pi^i$ intersects with at most two bundles of $\pi^j$. In other words, each of them is disjoint with a bundle of $\pi^j$, so the graph $G(\pi^i, \pi^j)$ contains two edges. By Corollary \ref{cor:2edges}, $I$ admits an MMS allocation.
	
	In the second case, no $\pi^i$ contains two 2-bundles, so each $\pi^i$ contains one 2-bundle and two 3-bundles. Let $I'$ be the instance obtained from $I$ by adding a dummy item with zero utility to every agent. For each $i$, let ${\pi^i}'$ be the allocation obtained from $\pi^i$ by adding the dummy item to the 2-bundle of $\pi^i$. Clearly, for each agent $i$, the allocation ${\pi^i}'$ is MMS for $i$ for the instance $I'$ and contains three 3-bundles. By Proposition \ref{prop:39}, $I'$ admits an MMS allocation. Since the addition of a dummy item does not affect the existence of MMS allocations, $I$ admits an MMS allocation as well.
	
\end{proof}

\subsection{$I$ contains a non-negative agent}

In this subsection, we prove Theorem \ref{thm:main1}, which constitutes the second part of Theorem \ref{thm:main}. Our proof uses the following preprocessing operation. Let $I$ be an instance of the fair division problem that contains an agent $i$ for whom $u_i^\text{MMS}(I) > 0$. We define the {\em mimicked instance} $I^{(i)}$ {\em of $I$ with respect to agent $i$} as the instance obtained from $I$ by replacing the utility function $u_j$ with $u_i$ for each agent $j$ such that $u_j^\text{MMS}(I) \leq 0$. Thus, every agent in the instance $I^{(i)}$ has positive MMS guarantee.

\begin{proposition}\label{prop:mimic}
	Let $I$ be an instance of the fair division problem and $i$ be an agent for whom $u_i^\text{MMS}(I) > 0$. If $I^{(i)}$ admits an MMS allocation, then $I$ does as well.
\end{proposition}
\begin{proof}
	For each agent $j$, we use $u_j$ and $u_j'$ to denote its utility function in $I$ and $I^{(i)}$, respectively. Let $\pi' = (\pi_1', \pi_2', \dots, \pi_n')$ be an MMS allocation for the instance $I^{(i)}$. We construct a new allocation from $\pi'$ as follows. For each agent $j \neq i$ such that $u_j(\pi_j') < u_j^\text{MMS}(I)$, we move all of the items in the bundle $\pi_j'$ to the bundle $\pi_i'$. We denote the allocation that results when no more items can be moved this way by $\pi = (\pi_1, \pi_2, \dots, \pi_n)$.
	
	We show that $\pi$ is MMS for the instance $I$ by showing that $u_j(\pi_j) \geq u_j^\text{MMS}(I)$ for each agent $j$. The first observation is that each agent $j \neq i$ with $u_j^\text{MMS}(I) > 0$ has the same utility function in $I$ and $I^{(i)}$, and hence has the same MMS guarantee. Thus, $u_j(\pi_j') \geq u_j^\text{MMS}(I)$ and $\pi$ allocates $\pi_j'$ to agent $j$. Therefore, the only agents that remain to be considered are agents $j \neq i$ such that $u_j^\text{MMS}(I) \leq 0$ and agent $i$.
	
	First consider the agents $j \neq i$ such that $u_j^\text{MMS}(I) \leq 0$. In this case, we have $u_j' = u_i$.	If $u_j(\pi_j') \geq u_j^\text{MMS}(I)$, then the items in the bundle $\pi_j'$ were not moved to the bundle $\pi_i'$ during the construction of $\pi$ from $\pi'$. Hence, $\pi$ allocates $\pi_j'$ to agent $j$ and satisfies agent $j$ because $u_j(\pi_j') \geq u_j^\text{MMS}(I)$. Otherwise, we have $u_j(\pi_j') < u_j^\text{MMS}(I)$ instead. In this case, the items in the bundle $\pi_j'$ were moved to the bundle $\pi_i'$ during the construction of $\pi$, so $\pi_j = \emptyset$ and we have $u_j(\pi_j) = 0 \geq u_j^\text{MMS}(I)$.
	
	It remains to consider agent $i$. In the allocation $\pi$, agent $i$ receives $\pi_i'$ together with every bundle $\pi_j'$ such that $u_j(\pi_j') < u_j^\text{MMS}(I)$. Since $\pi'$ is MMS for the instance $I^{(i)}$, we have $u_j'(\pi_j') \geq u_j^\text{MMS}(I^{(i)})$. This implies that $u_j \neq u_j'$ because otherwise, we would have $u_j^\text{MMS}(I) > u_j(\pi_j') = u_j'(\pi_j') \geq u_j^\text{MMS}(I^{(i)}) = u_j^\text{MMS}(I)$, a contradiction. Since $u_j \neq u_j'$, we have $u_j' = u_i$ by construction, so $u_j^\text{MMS}(I^{(i)}) = u_i^\text{MMS}(I^{(i)})$. Thus, $u_i(\pi_j') = u_j'(\pi_j') \geq u_j^\text{MMS}(I^{(i)}) = u_i^\text{MMS}(I^{(i)}) = u_i^\text{MMS}(I)$. On the other hand, we also have $u_i(\pi_i') \geq u_i^\text{MMS}(I^{(i)}) = u_i^\text{MMS}(I)$ because $\pi'$ is MMS for $I^{(i)}$. Thus, $\pi_i$ is a union of bundles that each has utility at least $u_i^\text{MMS}(I)$ to agent $i$. Since $u_i^\text{MMS}(I) > 0$ and $u_i$ is additive, we have $u_i(\pi_i) \geq u_i^\text{MMS}(I)$.
\end{proof}

\begin{theorem}\label{thm:main1}
	Let $I$ be an instance of the fair division problem. Suppose $m \leq n+5$ and $I$ contains a non-negative agent. Then, $I$ admits an MMS allocation.
\end{theorem}
\begin{proof}
	Without loss of generality, assume $m = n+5$. Since $I$ contains a non-negative agent, there exists an agent $i$ for whom $u_i^\text{MMS}(I) \geq 0$. First suppose that every agent has non-positive MMS guarantee. In particular, this implies that $u_i^\text{MMS}(I) = 0$. We show that the allocation $\pi = (\pi_1, \pi_2, \dots, \pi_n)$ which allocates every item to agent $i$ is an MMS allocation. Each agent $j \neq i$ receives no items, so $u_j(\pi_j) = 0 \geq u_j^\text{MMS}(I)$. On the other hand, let $\pi^i$ be an MMS allocation for agent $i$. By definition, we have $u_i(\pi_j^i) \geq u_i^\text{MMS}(I) = 0$ for each bundle $\pi_j^i$ of $\pi^i$. By the additivity of $u_i$, we have $u_i(\mathcal{M}) = \sum_{\pi_j^i \in \pi^i} u_i(\pi_j^i) \geq 0 = u_i^\text{MMS}(I)$. Thus, $\pi$ is an MMS allocation for the instance $I$.

	Otherwise, $I$ contains an agent $i$ with positive MMS guarantee. Let $I^{(i)}$ be the mimicked instance of $I$ with respect to agent $i$. By construction, $I^{(i)}$ is an instance in which every agent has positive MMS guarantee. Proposition \ref{prop:mimic} implies that if $I^{(i)}$ admits an MMS allocation, then $I$ does as well, so it suffices to show that $I^{(i)}$ admits an MMS allocation.
	
	We show that $I^{(i)}$ admits an MMS allocation by induction on the number $n$ of agents. If $n \leq 3$, then $I^{(i)}$ admits an MMS allocation by Proposition \ref{prop:2agents} or Theorem \ref{thm:38}. Assume $n \geq 4$ and that for all $n' < n$, any $n'$-agent $m'$-item instance $I'$ of the fair division problem with $m' \leq n' + 5$ in which every agent has positive MMS guarantee admits an MMS allocation.
	
	Suppose that for some agent $j$, there exists an allocation $\pi^j$ that is MMS for agent $j$ and contains no singletons. Since $u_j^\text{MMS}(I^{(i)}) > 0$, every bundle of $\pi^j$ has positive utility to agent $j$. In particular, $\pi^j$ contains no empty bundles, so every bundle of $\pi^j$ contains at least two items. Since $\pi^j$ contains exactly $n$ bundles, we have $2n \leq m \leq n+5$, so $(n, m)$ is one of $(4, 9)$ and $(5, 10)$. It follows that $\pi^j$ contains $n-1$ bundles of size 2. By Lemma \ref{lemma:atmost2reduction}, $I$ admits an MMS allocation or there exists a valid reduction rule $R$ that allocates to $k$ agents at least $k$ items for some $k > 0$. In the former case, we are done. In the latter case, let $I_R$ be the $n_R$-agent $m_R$-agent instance obtained from $I^{(i)}$ by applying $R$. We claim that $I_R$ admits an MMS allocation. Since $R$ preserves the MMS guarantees of each agent it does not delete from $I^{(i)}$, every agent in the instance $I_R$ is positive. Clearly, $m_R \leq n_R + 5$ and $n_R < n$. Thus, $I_R$ admits an MMS allocation by the induction hypothesis. By Proposition \ref{prop:validreduction}, the instance $I^{(i)}$ admits an MMS allocation.
	
	Otherwise, there is no agent $j$ for which there exists an allocation $\pi_j$ that is MMS for agent $j$ and contains no singletons. So, for each agent $i$, there exists an allocation that is MMS for agent $i$ and contains a singleton bundle. By Lemma \ref{lemma:1reduction}, there exists a valid reduction rule $R$ that allocates a single item to a single agent. let $I_R$ be the $n_R$-agent $m_R$-agent instance obtained from $I^{(i)}$ by applying $R$. Clearly, we have $m_R \leq n_R+5$ and every agent in $I_R$ has positive MMS guarantee because $R$ is a valid reduction rule. Thus, $I_R$ admits an MMS allocation by the induction hypothesis. By Proposition \ref{prop:validreduction}, the instance $I^{(i)}$ admits an MMS allocation.
\end{proof}
	
	\subsection{$I$ contains only chores agents}
	
	In this subsection, we prove Theorem \ref{thm:main2}, which constitutes the third part of Theorem \ref{thm:main}. This proof relies on the observation that instances with chores agents always have valid reduction rules as long as $m \leq 2n+1$, regardless of what other types of agents it contains.
	
	\begin{proposition}\label{prop:choresagent}
		Let $I$ be an instance with $m \leq 2n+1$. If there is a chores agent, then $I$ admits an MMS allocation or there exists a valid reduction rule that allocates to $k$ agents at least $k$ items for some $k > 0$.
	\end{proposition}
	\begin{proof}
		Without loss of generality, assume $m = 2n+1$ and agent 1 is a chores agent. If $I$ contains a non-negative agent, then Theorem \ref{thm:main1} implies $I$ admits an MMS allocation, so we are done. Thus, we assume every agent in $I$ is negative. Note that in this case, $I$ having SOP implies $o_m$ is a chore to every agent.
		
		For each agent $i$, let $\pi^i$ be an allocation that is MMS for $i$. If for some agent $i$, the allocation $\pi^i$ contains neither an empty bundle nor a singleton bundle, then the pigeonhole principle implies that $\pi^i$ contains $n-1$ bundles of size 2 and one bundle of size 3. By Lemma \ref{lemma:atmost2reduction}, $I$ admits an MMS allocation or there exists a valid reduction rule that allocates to $k$ agents at least $k$ items for some $k > 0$, so we are done.
		
		Otherwise, for every agent $i$, the allocation $\pi^i$ contains an empty bundle or a singleton bundle. For each $i$, let $\pi^i_\ell$ be a bundle in $\pi^i$ that is either empty or a singleton. We proceed to show that the reduction rule $R$ that allocates $o_m$ to agent $1$ is a valid reduction rule. Assume $o_m$ belongs to the bundle $\pi^1_1$ of $\pi^1$. Since agent 1 considers every item to be a chore, we have $u_1(o_m) \geq u_1(\pi^1_1) \geq u_1^\text{MMS}(I)$, so $R$ satisfies agent 1. It remains to show that $R$ preserves the MMS guarantee of every agent $i \neq 1$. Fix any agent $i \neq 1$. Suppose the bundle $\pi^i_\ell$ of $\pi^i$ is empty. Since $o_m$ is a chore to every agent, we have $u_i(o_m) \leq 0 = u_i(\pi^i_\ell)$. Otherwise, $\pi^i_\ell$ is a singleton. In this case, we again have $u_i(o_m) \leq u_i(\pi^i_\ell)$ because $I$ has SOP. By Proposition \ref{prop:1reduction_preserve}(1), $R$ preserves the MMS guarantee of agent $i$.
	\end{proof}
	
	\begin{theorem}\label{thm:main2}
		Let $I$ be an instance of the fair division problem. Suppose $m \leq n+5$ and $I$ contains only chores agents. Then, $I$ admits an MMS allocation.
	\end{theorem}
	\begin{proof}
		Let $I$ be an instance with $m \leq n+5$ and suppose $I$ contains only chores agents. If $n=2$ or $n=3$, then $I$ admits an MMS allocation by Proposition \ref{prop:2agents} or Theorem \ref{thm:38}. Assume $n \geq 4$ and that for all $n' < n$, any $n'$-agent $m'$-item instance $I'$ of the fair division problem with $m' \leq n'+5$ that contains only chores agents admits an MMS allocation.
		
		Since $n+5 \leq 2n+1$ for all $n \geq 4$, Proposition \ref{prop:choresagent} applies to $I$ and implies that $I$ admits an MMS allocation or there exists valid reduction rule $R$ that allocates to $k$ agents at least $k$ items for some $k > 0$. In the first case, we are done. Otherwise, let $I_R$ be the $n_R$-agent $m_R$-item instance obtained from $I$ by applying $R$. Clearly, $m_R \leq n_R+5$. Moreover, since reduction rules do not change the utilities of the remaining items to the remaining agents, a chores agent of $I$ that is not deleted by $R$ remains a chores agent in the instance $I_R$. In particular, $I_R$ contains only chores agents. By the induction hypothesis, $I_R$ admits an MMS allocation. By Proposition \ref{prop:validreduction}, the instance $I$ admits an MMS allocation.
	\end{proof}

	\section{Discussion}
	
	In this paper, we have studied the existence of MMS allocations for the mixed manna setting. Specifically, we showed that for an instance $I$ of the fair allocation problem, as long as $m \leq n+5$ and some auxiliary conditions hold (see Theorem \ref{thm:main}), then an MMS allocation exists. We also showed that these auxiliary conditions can be dropped if $n \leq 3$. Ultimately, it would be interesting to show whether or not these auxiliary conditions can still be eliminated for $n \geq 4$. Showing that they can be eliminated would imply a full generalization to the mixed manna setting.
	
	\begin{problem}\label{prob1}
		Let $I$ be an instance of the fair allocation problem such that $m \leq n+5$. Is it true that $I$ admits an MMS allocation?
	\end{problem}
	
	As an immediate consequence of Theorem \ref{thm:main}, the only remaining part of the puzzle is the instances with $n \geq 4$ that contain only negative mixed agents. We can assume they are mixed because negative non-mixed agents are chores agents, and we can use Proposition \ref{prop:choresagent} to handle chores agents.
	
	Handling negative mixed agents requires strategies of a different flavour than those used in the previous cases, and seems to be particularly difficult. A source of the difficulty is the fact that if an allocation $\pi^i$ is MMS for a negative mixed agent $i$, then it is possible that $\pi^i$ contains an empty bundle. As an example, consider an instance for which $n=4$, $m=9$, and $(u_1(o_1), u_1(o_2), \dots, u_1(o_9)) = (1, 1, 1, 1, 1, 1, -3, -3, -3)$. In this case, the only allocations that are MMS for agent 1 must have one empty bundle and three non-empty bundles, each of which contains two items of utility 1 and one item of utility -3. The presence of an empty bundle means that neither Lemma \ref{lemma:1reduction} nor Lemma \ref{lemma:atmost2reduction} applies. This is problematic because we can no longer use the valid reduction rules that we have obtained. It would be interesting to know if these remaining instances admit MMS allocations.
	
	\appendix
	
	\section{An Example}\label{appen:A}
	
	The following example illustrates a situation in which Proposition \ref{prop:delete_one_good} does not extend to the chores setting. We assume all agents have additive utility functions. We represent an instance $I$ using an $n \times m$ matrix $M(I)$ where each entry $m_{ij}$ takes the value $u_i(o_j)$. Since the utility functions are additive, the utility of any set of items can be inferred from the matrix.
	
	Consider the four instances $I, I^+_1, I^+_2, I^+_3$ represented by the matrices below. The MMS guarantee of agent 1 is shown on the right. Giving the fourth item to the second agent in any of $I^+_1, I^+_2, I^+_3$ yields $I$. Observe that the MMS guarantee of agent 1 decreases, stays the same, and increases if we do so in $I^+_1, I^+_2, I^+_3$, respectively.
	\begin{align*}
		M(I) &= \begin{bmatrix}
			-1 & -1 & -1
		\end{bmatrix} &u_1^\text{MMS}(I) &= -3 \\
		M(I^+_1) &= \begin{bmatrix}
			-1 & -1 & -1 & -1 \\
			-1 & -1 & -1 & -1
		\end{bmatrix} &u_1^\text{MMS}(I^+_1) &= -2 \\
		M(I^+_2) &= \begin{bmatrix}
			-1 & -1 & -1 & -3 \\
			-1 & -1 & -1 & -3
		\end{bmatrix} &u_1^\text{MMS}(I^+_2) &= -3 \\
		M(I^+_3) &= \begin{bmatrix}
			-1 & -1 & -1 & -4 \\
			-1 & -1 & -1 & -4
		\end{bmatrix} &u_1^\text{MMS}(I^+_3) &= -4
	\end{align*}
	
	\section{Proof of Proposition \ref{prop:39}}\label{appen:B}
	
	We include below the proof of Proposition \ref{prop:39} due to \citet{feige2021tight} for completeness. For each agent $i \in [3]$, we fix an allocation $\pi^i$ that is MMS for $i$.
	
	\propthreenine*
	\begin{proof}
		Assume $I$ has same-order preference. For each agent $i$, let $\pi^i = (\pi^i_1, \pi^i_2, \pi^i_3)$ be an allocation that is MMS for $i$.
		
		Suppose that for a pair of agents $i, j$, there exist two bundles $\pi^i_a$ and $\pi^j_b$ that are identical. Without loss of generality, assume $\pi^1_1 = \pi^2_1$. By the additivity of $u_1$, agent $1$ is satisfied with at least one of the other two bundles of $\pi^2$. Hence, agent $1$ is satisfied with two bundles of $\pi^2$ (including $\pi^2_1$), and agent $2$ is satisfied with three bundles of $\pi^2$ by definition. Thus, using the picking order $3, 1, 2$ results in an MMS allocation.
		
		Otherwise, suppose that for a pair of agents $i, j$, there exist two bundles $\pi^i_a$ and $\pi^j_b$ that share exactly two items, that is, they differ in exactly one item. Without loss of generality, assume this is true for bundles $\pi^1_1$ and $\pi^2_1$. Since $I$ has same-order preference, every agent agrees on which of these two bundles is at least as valuable as the other. Assume $u_\ell(\pi^2_1) \geq u_\ell(\pi^1_1)$ for every agent $\ell \in [3]$. Since $u_2(\pi^2_1) \geq u_2(\pi^1_1)$, the additivity of $u_2$ implies that $u_2(\pi^1_2 \cup \pi^1_3) \geq u_2(\pi^2_2 \cup \pi^2_3) \geq 2u_2^\text{MMS}(I)$. This allows us to assume $u_2(\pi^1_2) \geq u_2^\text{MMS}(I)$ without loss of generality.
		
		If $u_3(\pi^1_1) \geq u_3^\text{MMS}(I)$, giving $\pi^1_1$ to agent 3, $\pi^1_2$ to agent 2, and $\pi^1_3$ to agent 1 results in an MMS allocation. Otherwise, $u_3(\pi^1_1) < u_3^\text{MMS}(I)$. In this case, we first give $\pi^1_1$ to agent 1, satisfying them. Without loss of generality, assume the item $o_x \in \pi^1_1 \setminus \pi^2_1$ is in the bundle $\pi^2_2$ of $\pi^2$. Agent 2 is able to partition the remaining items into two bundles that both satisfy agent 2, by replacing the item $o_x$ in $\pi^2_2$ with the item $o_y \in \pi^2_1 \setminus \pi^1_1$ (which is at least as good as $o_x$). One of these two resulting bundles must satisfy agent 3 because $\frac{u_3(\mathcal{M}) - u_3(\pi^1_1)}{2} \geq u_3^\text{MMS}(I)$. Thus, an MMS allocation can be obtained by letting agent 3 pick first from these two bundles.
		
		Finally, we consider the remaining case where the intersection between each pair of bundles $\pi^i_a, \pi^j_b$ where $i \neq j$ contains exactly one item. Without loss of generality, assume the least valuable item $o_9$ is in $\pi^i_1$ for each $i \in [3]$. Except for containing $o_9$, the three bundles $\pi^i_1$ are disjoint. Hence, there are two items $o_x$ and $o_y$ that are in none of these bundles, each having utility at least that of $o_9$ to every agent. Giving $\pi^3_1$ to agent 3, $(\pi^2_1 \setminus \{o_9\}) \cup \{o_x\}$ to agent 2, and $(\pi^1_1 \setminus \{o_9\}) \cup \{o_y\}$ to agent 1 results in an MMS allocation.
	\end{proof}



\begin{ack}
	We acknowledge the support of the Natural Sciences and Engineering Research Council of Canada (NSERC), funding reference number RGPIN-2022-04518. We also thank Halvard Hummel for constructive comments on the organization of this paper.
\end{ack}



\bibliography{mybibfile}

\end{document}